\documentclass[conference,letterpaper]{IEEEtran}
\IEEEoverridecommandlockouts 
\usepackage[utf8]{inputenc} 
\usepackage[T1]{fontenc}
\usepackage{url}              
\usepackage[cmex10]{amsmath}
\usepackage{amssymb}
\usepackage{amsthm}
\usepackage{algorithmic}
\usepackage{textcomp}
\usepackage{xcolor}
\def\BibTeX{{\rm B\kern-.05em{\sc i\kern-.025em b}\kern-.08em
    T\kern-.1667em\lower.7ex\hbox{E}\kern-.125emX}}
\usepackage{graphicx}
\usepackage{subfiles}
\usepackage[inline]{enumitem}

\usepackage{array} 
\usepackage{url}
\usepackage{epstopdf}
\usepackage{tikz}
\usetikzlibrary{automata,arrows,positioning,calc,shapes,chains, patterns}
\usepackage{pifont}
\usepackage{mathtools}
\usepackage{texdef2015}
\usepackage[noadjust]{cite}
\allowdisplaybreaks
\newtheorem{proposition}{Proposition}

\newcommand{\Mcal}{\mathcal{M}}

\renewcommand{\vec}[1]{\begin{bmatrix} #1\end{bmatrix}}

\newlength{\swwidth}

\newcommand{\al}{\alpha}

\newcommand{\Valp}[1]{V(#1)}

\newcommand{\Valpn}[2]{V_{#1}(#2)}
\newcommand{\Jt}[2]{\Jtil_{#1}(#2)}
\newcommand{\Jvt}[1]{\Jtil(#1)}

\newcommand{\fal}[1]{f_{\alpha}(#1)}

\newcommand{\gs}{g}
\newcommand{\yzero}{Y_0}
\newcommand{\yzeros}{Y_0^*}
\newcommand{\jzero}{J_0}
\newcommand{\gzero}[1]{g_0(#1)}
\DeclareSymbolFont{matha}{OML}{txmi}{m}{it}
\DeclareMathSymbol{\varv}{\mathord}{matha}{118}

\newtheoremstyle{myremark}
  {\topsep}
  {\topsep}
  {\normalfont} 
  {}
  {\itshape} 
  {.} 
  {.5em} 
  {}

\theoremstyle{myremark}

\begin{document}
\title{Efficient and Timely Memory Access}
\author{%
 \IEEEauthorblockN{Vishakha Ramani, Ivan Seskar, Roy D. Yates}
 \IEEEauthorblockA{WINLAB,
                   Rutgers University\\
                   Email: \{vishakha, seskar, ryates\}@winlab.rutgers.edu}
}
\maketitle

\begin{abstract}
    This paper investigates the optimization of memory sampling in status updating systems, where source updates are published in shared memory, and reader process samples the memory for source updates by paying a sampling cost. We formulate a discrete-time decision problem to find a sampling policy that minimizes average cost comprising age at the client and the cost incurred due to sampling. 
    We establish that an optimal policy is a stationary and deterministic threshold-type policy, and subsequently derive optimal threshold and the corresponding optimal average cost.
\end{abstract}

\section{Introduction}
This work examines status updating systems in which sources generate time-stamped status updates of a process of interest, and these updates are stored/written in a memory system.
A reader fulfills clients' requests for these  updates by reading/sampling from the memory. The asynchronous nature of reader-writer interactions within memory systems introduces significant challenges. In particular, the readers' memory accesses should be optimized for timely processing of source updates as the reader becomes aware of fresher updates in the memory only when it chooses to query the memory. Furthermore, the memory access process must be regulated by a synchronization method between readers and writers to avoid race conditions. 


The primary question 
in this paper 
is when should the reader sample the memory.
Typically, there is a cost associated with memory sampling, and this cost structure varies between systems.
In systems with substantial object sizes, retrieving and locally copying objects incurs a high cost, while querying for timestamps remains relatively inexpensive. 
In contrast, there are systems where memory contains smaller objects, 
and the cost of retrieval is comparable to the cost of a timestamp query. 
These are systems where queries are sent to a distant database, with the cost being the latency associated with the query.

In this work, we focus on former class of systems where the reader knows the freshness of object in the memory by virtue of inexpensive timestamp retrievals. 
However, due to longer read times, denoted by high sampling costs, 
the Reader must decide if sampling is justified compared to age reduction obtained after sampling.



\subsection{Related Work}
Prior research on timely memory access has explored issues related to the impact of different synchronization primitives on the timely retrieval of stored data items. Particularly, \cite{ramani-cyINFOCOM, ramani-cy-AoI-INFOCOM} examined the impact of lock-based and lock-less synchronization primitives in the context of a timely packet forwarding application.

In \cite{ramani-sy-wiopt2023}, the authors addressed the timely processing of updates stored in memory from multiple sources.
They proposed a system where the reader samples memory as a renewal point process, showing that a lazy sampling policy was optimal. The rationale was that lazy sampling mitigated negative impacts of high variance in client processing times.
This work differs from \cite{ramani-sy-wiopt2023} as it no longer assumes renewal sampling and introduces a cost for sampling. 
Nevertheless, the primary objective remains consistent: to identify an optimal reading policy maximizing source update timeliness with the client.

We note that the concept of timely memory sampling, wherein the reader incurs a cost for sampling for age reduction, shares similarities with research focused on managing access for multiple users within a communication channel. Various studies in the AoI literature have explored Whittle's index-based transmission scheduling algorithms \cite{yupinwhittles, yupin-tomc2020, kadotascheduling, Maatouk-ToWC-Whittles, Sun-BK-Whittles, Kadota-SM-ToN2019, tripathi2019whittle, maatouk2020asymptotically, sun-whittles-infocom, Jiang-BK-ITC2018}, where the scheduling problem is decomposed into multiple independent subproblems. Within each subproblem, an additional cost $(C)$ is associated with updating the user. 

However, there is a conceptual difference in the cost associated with the decoupled problem and this study. In a Whittle index policy, the
minimum cost that makes 
both actions --- updating a user or idling --- equally desirable 
is used as a mechanism to choose one of the many users. 
In this work, we enforce an explicit cost of accessing the memory, and we study the trade-offs observed with age and memory access by varying this system parameter. However, it is not a mechanism designed to distinguish between users.

\subsection{Contributions and Paper Outline}
This paper investigates the relation between sampling costs and Age-of-Information \cite{Kaul-YG-infocom2012}. 
In section~\ref{sec:system-model}, we formulate our problem as a Markov Decision Process (MDP) with the goal of minimizing average cost comprising age at the client and the cost incurred due to sampling.
In section~\ref{sec:avg-cost-opt}, we establish that an optimal policy of the MDP is a stationary and deterministic threshold-type policy. We then derive optimal threshold and the optimal average cost by exploiting the structure of optimal policy. Finally, section~\ref{sec:num-evaluation} presents numerical evaluation on average cost against system parameters.

\section{System Model} \label{sec:system-model}
In this work, we focus on a  class of systems (see Fig.~\ref{fig:mem-samp-cartoon}) where a Writer writes the time-varying data received from the source into the memory, and a Reader samples the memory on behalf of a client. 
The client can  be either the same entity as the Reader, running as a single process, or a separate process. 
Additionally, while the system will have many sources, our focus will be on the memory that tracks  the status of a single process of interest. 
Even in this seemingly straightforward scenario, not previously explored in the AoI literature, optimizing AoI presents non-trivial challenges.

We consider a discrete-time slotted system with slots labelled $t=0,1,2, \ldots$. 
The modeling details of writing and reading processes are discussed below.

\subsection{Writing source updates to memory}
We assume the Writer commits/writes fresh (age zero) source updates to memory at the end of each slot with probability $p$,  independent from slot to slot. These source updates generate the age process $x(t)$ in the memory.

In practice, the write time will be non-negligible.
However, our focus in this work is not on systems where writing to the memory is the bottleneck process. Instead, our primary interest lies in examining the delays associated with reading and processing of source updates. Note that in the event that these writes do require time $\tau>0$, $x(t)$ and the update age process at the client will be shifted by $\tau$.

\begin{figure}[t]
    \centering
    \includegraphics[width=\linewidth]{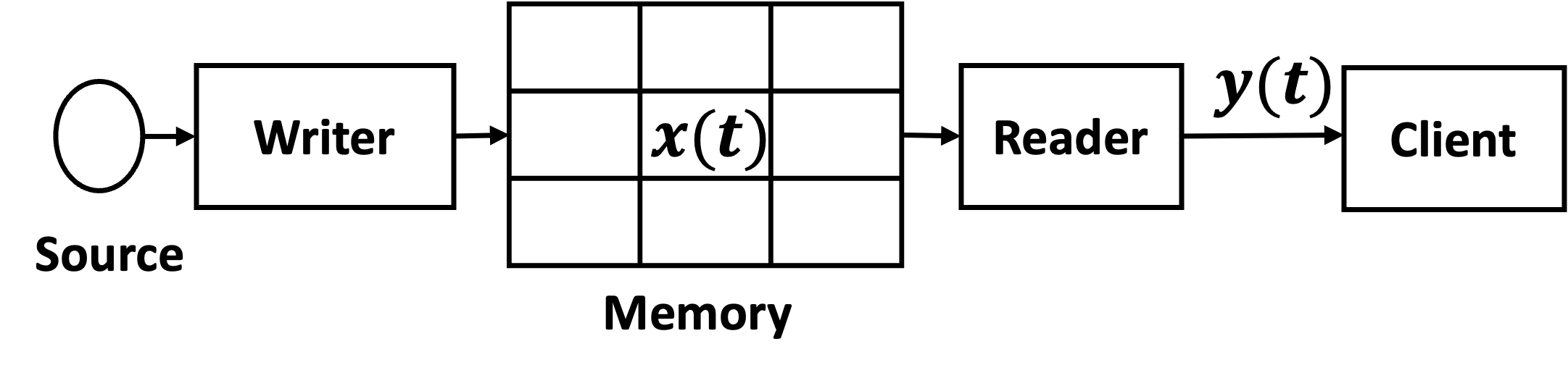}
    \caption{A writer updates memory based on the update received from source. A client requests the Reader process to read the source updates from the memory. The source update publication in the memory generates age process $x(t)$, and the update sampling by the Reader generates age process $y(t)$ at the client input.}
    \label{fig:mem-samp-cartoon}
\end{figure}

\subsection{Reading source updates from memory}
At each time slot, the Reader determines whether to access the memory and read a source update. The update in memory is read over a period of a slot, and the reader gets the data at the end of the slot. 
Notably, this model aligns with the Read-Copy-Update (RCU) \cite{mckenney1998read, MckenneyRCU2001} memory access paradigm, where a new update can be written in slot 
$t$
while the Reader is in the process of reading the current update in the same slot.

The Reader generates  an age process $y(t)$  at the input to the client that is a sampled version of source update age process $x(t)$ in the memory. Hence, we say the Reader {\em samples} the updates in the memory.

The state-dependent action $a(t)$ selected by the Reader at time slot $t$ determines whether the Reader remains idle ($a(t) = 0$) or performs a read operation ($a(t) = 1$). 
We consider a scenario where a non-negative fixed cost $c$ is associated with reading the memory during each time slot. Ideally, the Reader aims to minimize $y(t)$, which means it would prefer to read in every slot to stay close to the age process $x(t)$.
 However, this comes at the cost of paying the sampling cost $c$. If the Reader samples too frequently, it might end up with the same update, resulting in no age reduction but incurring a penalty for sampling. On the contrary, if it reads too infrequently, the age at the client input increases.

In this work, we assume that the Reader is notified when an update is published in the memory, enabling the Reader to know the update age in the memory. Based on the system state, 
the Reader implements a scheduling scheme 
that minimizes the average cost $\E{y(t) + ca(t)}$. To address this, we model our problem as a Markov Decision Process (MDP).

\subsection{Markov Decision Process Formulation}
In the context of our MDP model, denoted with $\Mcal$ from here on, 
the following four 
components 
make up the structure: 
\begin{itemize}
    \item States: We denote the set of possible system states by $S$ which does not vary with time. State $s(t) \in S$ is a tuple $(x(t), y(t))$, where at the start of a time slot, $x(t) \in \{0, 1,2,\ldots\}$ is the age of the update in the memory, and $y(t)\in \{1,2,3,\ldots\}$ is the age of sampled source updates at the client. 
    Notice that $S$ is a countably infinite set since age is unbounded. 
    \item Action: Let $a(t) \in A = \{0, 1\}$ denote the action taken in slot $t$ indicating Reader's decision, where $a(t) = 1$ if Reader decides to read and $a(t) = 0$ if idle.
    \item Transition Probabilities: Letting $\pbar=1-p$, when $a(t) = 1$, the transition probability from state $s=(x,y)$ to state $s' \in S$ is
\begin{subequations}\eqnlabel{transproba}
\begin{IEEEeqnarray}{rCl}
\prob{s' \mid s = (x,y), a = 1} 
&=& \begin{cases}
p &\hspace{-0.44em} s'= (0,x+1),\\ 
\pbar &\hspace{-0.44em} s'= (x+1,x+1). 
\end{cases}
\IEEEeqnarraynumspace       
\eqnlabel{transproba1}
\end{IEEEeqnarray}
And when $a(t) = 0$, the transition probability is
\begin{IEEEeqnarray}{rCl}
\prob{s' \mid s = (x,y), a = 0} 
&=& \begin{cases}
p &\hspace{-0.44em} s'= (0,y+1),\\ 
\pbar &\hspace{-0.44em} s'= (x+1,y+1).
\end{cases}
\IEEEeqnarraynumspace       
\eqnlabel{transproba0}
\end{IEEEeqnarray}
\end{subequations}
\item Cost: The cost $C(s(t);a(t))$ incurred in state $s(t)$ in time slot $t$ under action $a(t)$ is defined as $C(s(t) = (x,y); ~a(t) = a) \coloneqq y + ca$.
\end{itemize}
Let $\pi: S \rightarrow A$ denote a policy that for each state $s(t) \in S$ specifies an action $a(t) = \pi(s(t)) \in A$ at slot $t$.
The expected average cost under policy $\pi$ starting from a given initial state at $t=0$, $s(0) = (x,y)$, is defined as:
\begin{equation}
    g_\pi(x,y) = \limsup_{{T \to \infty}} \frac{1}{T}
    \mathbb{E}_\pi\left[\sum_{t=0}^{T-1}(y(t) + ca(t))\right]
\end{equation}
We say that policy $\pi^*$ is average-cost optimal if
$g_{\pi^*}(s) = \inf g_\pi(s)$ for every $s\in S$. 
We focus on the case where for some constant $\gs$, $g_{\pi^*}(s) = \gs$ for all $s \in S$. Thus, the problem is to obtain $\pi^*$ such that $g = g_{\pi^*}(s) = \inf g_\pi(s)$ for every $s \in S$.


Our cost minimization problem falls within the category of average cost minimization problems. Given that the age can grow unbounded, both the number of states and the cost in each stage are countably infinite. In such MDPs, the existence of an optimal policy, whether stationary or non-stationary, is not guaranteed \cite[Chap~5]{ross83sdp}. 
Notably, even the existence of an optimal stationary policy may not hold, while an optimal non-stationary policy might exist \cite{ross1970applied}.

Proving the existence of an optimal average cost stationary policy is not an immediate goal in this paper and we defer this discussion to later in Section~\ref{sec:opt-pol-exist}. There, we draw upon results from \cite{sennott88}, which provides conditions ensuring the existence of an expected average cost optimal stationary policy. We verify that these conditions hold for our problem. In the subsequent section, we derive results regarding the structure of the optimal policy under the assumption that the optimal policy exists and the relative cost Bellman's equation is valid.

\section{Characterization of cost optimality}
\label{sec:avg-cost-opt}
\subsection{Discounted Cost}
We begin by introducing the $\alpha$-discounted version of the problem. Recall that the state for MDP $\Mcal$ is a tuple $s = (x,y)$, and $a \in \{0,1\}$. Then using \eqnref{transproba},
the discounted cost Bellman's optimality equation for $\Mcal$ is given by
\begin{align}
    \Valp{x,y} &= \min\{ y + \al\left(p\Valp{0,y+1} + \pbar\Valp{x+1, y+1}\right), \nn
    & y+c + \al(p\Valp{0,x+1}+\pbar\Valp{x+1,x+1})\}.
    \eqnlabel{valfunc}
\end{align}
Here, the first term of $\min$ corresponds to the reader staying idle ($a=0$), and the second term corresponds to the reader sampling ($a=1$).
The action that is a minimizer of \eqnref{valfunc} is referred to as the {\em $\al$-optimal action} and the resulting policy $\pi^*_\al$ is referred to as the {\em $\al$-optimal policy}.

We define the value iteration $\Valpn{n}{s}$ by
$\Valpn{0}{s} = 0, \forall s \in S$, and, for any $n > 0$,
\begin{align}
&\Valpn{n+1}{x,y} = 
\scalebox{0.95}{$
\min\{ y + \al\left(p\Valpn{n}{0,y+1} + \pbar\Valpn{n}{x+1, y+1}\right), 
$}
\nn
    &\qquad y+c + \al(p\Valpn{n}{0,x+1}+\pbar\Valpn{n}{x+1,x+1})\},
    \eqnlabel{valiteropteqn}
\end{align}
For non-negative costs, it is evident that $\Valpn{n}{s} \leq \Valpn{n+1}{s}$. 
It then follows  from \cite[Theorem~4.2, Chapter III]{ross83sdp} that
\begin{equation}
    \lim_{n \to \infty} \Valpn{n}{s} = \Valp{s}, \qquad s \in S.
\end{equation}
We now state properties of the value function. 
\begin{proposition}(Monotonicity): \label{prop:mp-mono}
    The value function $\Valp{x,y}$ is non-decreasing in both $x$ and $y$.
\end{proposition}
\noindent The proof, which uses mathematical induction on \eqnref{valiteropteqn}, is straightforward and has been omitted. 
From here, proofs, if not shown, appear in the Appendix.
\begin{proposition} \label{prop:mp-thresh}
If the $\al$-optimal action is to sample in $(x,y)$, then the $\al$-optimal action is to sample in every $(x,y')$ with $y' \geq y$.
\end{proposition}
\noindent
Another version of this proposition asserts that if the $\al$-optimal action is to sample in state $(x,y)$ at stage $n$, then it is also optimal to sample in every $(x,y')$ with $y'\geq y$ at stage $n$. The proof employing the value iteration \eqnref{valiteropteqn} is omitted as it is similar to that of Proposition~\ref{prop:mp-thresh}.  

\begin{proposition}\label{prop:concavity}(Concavity):
For a fixed $x$, $\Valp{x,y+1} - \Valp{x,y}$ is non-increasing in $y$.
\end{proposition}
The intuitive structure of the optimal policy is that with knowledge of the age in the memory, the Reader should refrain from sampling if the reduction in age doesn't justify the sampling cost. 
To further characterize this intuition, we introduce the following proposition. 
\begin{proposition} \label{prop:policy-hypo}
    If the $\al$-optimal action in state $(x,y)$ is to idle, then the $\al$-optimal action in states $(x+i,y+i), \forall i\geq 1$ is to stay idle.
\end{proposition}
Specifically, when the memory is freshly updated, the Reader must assess whether sampling is worthwhile. If it opts against sampling initially, it should consistently abstain from sampling in subsequent slots until the memory undergoes another update, as the age reduction remains constant in the absence of changes. 
In terms of the MDP $\Mcal$, this concept translates to making a decision in the state $(0,y)$. If the optimal decision is not to sample in $(0,y)$, then the Reader should consistently refrain from sampling in states $(1, y+1)$, $(2, y+2)$, and so on.

\subsection{Average Cost Optimality}
Since the conditions of~\Thmref{opt-policy-exist} (in section~\ref{sec:opt-pol-exist}) 
hold, the average-cost optimal policy $\pi^*$ is the limit point of $\al$-optimal policies $\pi^*_\al$ with $\al \to 1$ \cite[Lemma]{sennott88}.
Therefore, Propositions \ref{prop:mp-thresh} and \ref{prop:policy-hypo} are sufficient to provide the structure of average cost optimal policy. Specifically, Propositions \ref{prop:mp-thresh} and \ref{prop:policy-hypo} imply that there exists a threshold $\yzero$ such that it is optimal to sample in $(0,y)$ for every $y \geq \yzero$ and idle otherwise. 

At this point, it is important to mention the set of feasible states under $\pi^*$. 
With $\yzero=1$, the optimal policy dictates sampling in every state $(0,y)$ with $y\geq 1$. Upon sampling in $(0,1)$, the system transitions to feasible states, specifically $\{(0,1), (1,1)\}$. In state $(1,1)$, a close examination of Bellman's equation \eqnref{valfunc} reveals that it is optimal to idle. Therefore, the set of possible states when choosing to idle in $(1,1)$ becomes $\{(0,2), (2,2)\}$.
Subsequent transitions follow a pattern where sampling in $(0,y)$ leads
to states $\{(0,1), (1,1)\}$, and choosing to idle in states $(i,i)$ with $i \in \mathbb{N}$ resulting in $\{(0,i+1), (i+1, i+1)\}$.

In scenarios where $\yzero>1$, optimality dictates idling in $(0,y)$ with $y < \yzero$, prompting the system to transition to states $\{(0,y+1), (1,y+1)\}$. The subsequent action in $(0,y+1)$ hinges upon whether $y+1 < \yzero$. If $y+1 \geq \yzero$, the system resets, transitioning to either $(0,1)$ or  $(1,1)$; conversely, if $y+1 < \yzero$, the system perpetuates a structure akin to that observed in state $(0,y)$. Conversely, if the system transitions to $(1,y+1)$, idling in $(1,y+1)$ is optimal. The resulting permissible states from this point include $\{(0,y+2), (2, y+2)\}$, and this pattern repeats.
We summarize this set of feasible states for the optimal policy in the following  proposition. 
\begin{proposition}
    For MDP $\Mcal$, under the optimal policy $\pi^*$ with threshold $\yzero$, the set of feasible states is
    \begin{equation}
S^* =\{(0, y) \mid y \in \mathbb{N}\} \cup \{(x, y) \mid x \geq 1 \text{ and } y - x < \yzero\}.
    \end{equation}
\end{proposition}




To determine the optimal threshold for an optimal policy $\pi^*$, we employ the relative cost Bellman's equation 
\begin{align}
    \gs + f(x,y) &= \min\{ y + pf(0,y+1) + \pbar f(x+1,y+1), \nn
    & y+c + p f(0,x+1)+ \pbar f(x+1,x+1)\}.
    \eqnlabel{rc-bellmans}
\end{align}
Here, $\gs$ denotes the optimal average cost, and $f(x,y)$ represents the relative cost-to-go function. Our objective is to identify relative cost-to-go function $f(x,y)$ for $(x,y) \in S^*$, facilitating the determination of the optimal threshold and, consequently, the optimal average cost. 

\begin{proposition} \label{prop:rel-cost-func}
Defining $(0,1)$ as the reference state with $f(0,1) = 0$, the relative cost functions satisify:
\begin{enumerate}[label={(\roman*)},itemindent=1em]
    \item $f(0,\yzero+1) - f(0,\yzero) = 1.$
        \label{prop:rc-diff-y0y0pl1}
        
    \item For any $x \geq 0$,
        \begin{equation}
            f(x, \yzero-1) = \frac{1}{p}(\jzero + \frac{\pbar}{p}) - 1, \eqnlabel{rc-y0min1}
        \end{equation}
        where $\jzero = \yzero - g + pf(0,\yzero)$. 
        
    \item For every $y < \yzero$, 
        $f(0,y) = f(1,y) \ldots = f(y,y).$
        \label{prop:rc-leq-y0}
        
    
        
    \item  When $\yzero > 1$, 
        $f(0,\yzero) = \yzero-\gs + c.$
        \label{prop:rc-y0}

        
\end{enumerate}
\end{proposition}
\noindent We use Proposition~\ref{prop:rel-cost-func} to derive the optimal threshold.




\begin{lemma} \label{lemma:gy0exp}
    As a function of the threshold $\yzero$, the average cost is
    \begin{equation}
        \gzero{\yzero} = \frac{1}{2}\left(\frac{1}{p} + \yzero + \frac{2cp + \pbar/p}{p\yzero+\pbar}\right).
        \eqnlabel{gy0}
    \end{equation}
\end{lemma}
\begin{theorem} \thmlabel{opt-thresh}
    The optimal threshold $\yzeros$ associated with optimal policy $\pi^*$ for MDP $\Mcal$ is $\yzeros=\ceiling{Y'}$ where 
    \begin{equation}
        Y' = \sqrt{2c+\left({1}/{p} - {1}/{2}\right)^2 } - 
        ({1}/{p} - {1}/{2}). 
        \eqnlabel{opt-thresh-eqn}
    \end{equation}
\end{theorem}
\begin{proof}
    It follows from \eqnref{gy0} and some algebra that     \begin{equation}
        \gzero{Y_0} - \gzero{Y_0+1} =\frac{-p^2}{2}\bracket{\frac{Y_0^2+(2/p-1)Y_0-2c}{(pY_0+\pbar)(pY_0+1)}}
        \eqnlabel{g0-change}
    \end{equation}
    We define $Q(Y_0)\equiv Y_0^2 + Y_0\left({2}/{p} - 1\right) - 2c$ and we observe that $Y'$ in \eqnref{opt-thresh-eqn} is the only positive root of  $Q(y)$. Further $Q(Y_0)>0$ for $Y_0>Y'$. It then follows from \eqnref{g0-change} that  $g_0(\floor{Y'})\ge g_0(\ceiling{Y'})$ and that $g_0(\ceiling{Y'}), g_0(\ceiling{Y'}+1),\ldots$ is a non-decreasing sequence. 
\end{proof}
\Thmref{opt-thresh} provides an explicit expression for the optimal threshold $\yzeros$. However, evaluating the optimal average cost $g = \gzero{\yzeros}$ using \eqnref{gy0} doesn't directly show the relationship between system parameters $c$ and $p$.
The following lemma provides a close approximation to the optimal average cost and captures the impact of these key system parameters on the average cost.

\begin{lemma} \label{lemma:glowerbound}
    The optimal average cost satisfies
\begin{equation}
\gs \geq {1}/{2} + \sqrt{2c + 1/p^2 -1/p}.
\eqnlabel{glowerbound}
        \end{equation}
\end{lemma}
\begin{proof}
Note that $g=\min_{\yzero\in\mathbb{N}} g_0(\yzero) \ge \min_{y\in\mathbb{R}^+} g_0(y)$. To minimize $g_0(y)$ over positive reals, we set  $dg_0(y)/dy=0$, yielding
$y=\tilde{Y}_{0}^* = -\pbar/p + \sqrt{2c +\pbar/p^2}$. This yields  $g\ge g_0(\tilde{Y}_0^*)$, which is the lower bound \eqnref{glowerbound}. 
\end{proof}

\begin{figure}[t]
    \centering
    \includegraphics{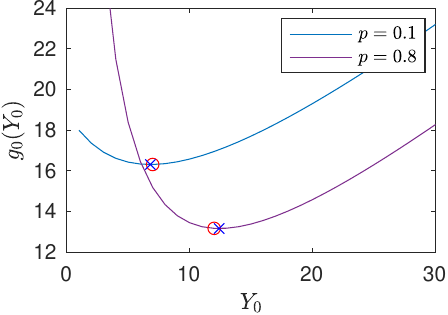}
    \caption{Plot of average cost $\gzero{\yzero}$ as a function of threshold $\yzero$ with sampling cost $c=80$. Here, \textcolor{magenta}{\large $\circ$} is the true optimal cost $\gzero{\yzeros}$, and  \textcolor{blue}{\large $\times$} is the approximate optimal average cost $\gzero{\tilde{Y}_0^*}$.}
    \label{fig:gy0plot}
\end{figure}

\section{Numerical Evaluation}
\label{sec:num-evaluation}
Fig.~\ref{fig:gy0plot} shows how the average cost $\gzero{\yzero}$, given by \eqnref{gy0}, changes with threshold $\yzero$. Initially, as $\yzero$ increases, the average cost decreases. This is because a low threshold leads to excessive sampling, incurring costs without much age reduction, resulting in a higher average cost. As $\yzero$ increases further, the cost of sampling approaches the gain in age reduction. However, setting $\yzero$ too high delays memory access, increasing client age and consequently the average cost. Fig.~\ref{fig:gy0plot} highlights the existence of an optimal threshold where the cost of sampling justifies the age reduction.

Fig.~\ref{fig:optimaly0} illustrates the optimal threshold $Y_0^*$
as a function of the source update probability $p$. 
We observe that the optimal threshold increases with $p$. When the Reader is required to make a decision in a given slot, it assesses both the age at the client and the age in the memory. These evaluations contribute to determining the potential age reduction vs the cost of sampling. 
In scenarios where the client's update is deemed sufficiently recent, the Reader may choose to skip sampling. This decision is influenced by a higher probability ($p$) of obtaining a more recent update soon, that will perhaps be worth sampling. 

Fig.~\ref{fig:optimalg} compares the optimal average cost $g$ with the lower bound from \eqnref{glowerbound}. 
The tightness of the lower bound is evident, as it closely aligns with the curve of the optimal average cost.
The plot also demonstrates that $g$ increases with sampling costs $c$.
This suggests that while designing the cost structure, the cost should be sufficiently high but not excessively so. 
Moreover, $g$ decreases as the update probability $p$ increases.
This is intuitive, as frequent memory updates increase the likelihood of the Reader receiving a fresh update when it samples, thereby reducing the age at the client.

\begin{figure}[t]
    \centering
    \includegraphics{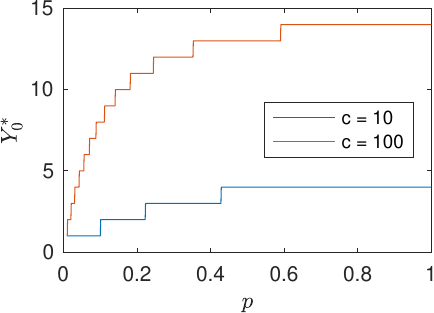}
    \caption{Plot of optimal threshold $\yzero^*$ as a function of probability $p$ of source update publication in a slot, with a fixed sampling cost $c$.}
    \label{fig:optimaly0}
\end{figure}

\section{Stationary average cost optimal policy}
\label{sec:opt-pol-exist}

In this section, we verify that the average cost optimality equation for MDP holds for $\Mcal$. To get started, we need the following result.
\begin{lemma}
\label{lemma:dtmcpolicy}
    Under the deterministic stationary policy of reading in every slot, the system exhibits an irreducible, ergodic Markov Chain, with expected cost $M(x,y)$ of first passage from state $s=(x,y)$ to $(0,1)$ satisfying
\begin{equation}     M(x,y)\le \frac{1+p}{p^2}(c+y) + \frac{3}{2p^3}. \eqnlabel{mxybound}
\end{equation}
\end{lemma}
\noindent We now employ Lemma~\ref{lemma:dtmcpolicy} in verifying the conditions of the following theorem. 
\begin{theorem} \thmlabel{opt-policy-exist} 
   \cite[Theorem]{sennott88} If the following conditions hold for MDP $\Mcal$:
    \begin{enumerate}
    \item For every state $s$ and discount factor $\al$, the quantity $\Valp{s}$ is finite,
    \item $\fal{s}\coloneqq \Valp{s} - \Valp{0}$ satisfies $-N \overset{(a)}\leq \fal{s} \overset{(b)}\leq M(s)$, where $M(s) \geq 0$, and
    \item For all $s$ and $a$, $\sum_{s'} \mathbb{P}_{s,s'}(a)M(s') < \infty$,
    \end{enumerate}
    then there exists a stationary policy that is average cost optimal for MDP $\Mcal$. Moreover, for $\Mcal$, there exists a constant $\gs = \lim_{\al \to 1}(1-\al)\Valp{s}$ for every state $s$, and a function $f(s)$ with $-N \leq f(s) \leq M(s)$ that solve relative-cost Bellman's equation, 
    \begin{equation}
        \gs + f(s) = \min_a\{C(s;a) + \sum_{s' \in S} \mathbb{P}_{s,s'}(a) f(s')\}.
    \end{equation}
\end{theorem}

\begin{figure}[t]
    \centering
    \includegraphics{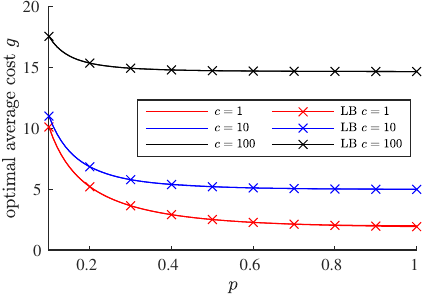}
    \caption{Comparison of optimal average cost $g$ and the corresponding lower bound (LB) as a function of probability $p$ of source update publication in a slot, with a fixed sampling cost $c$.}
    \label{fig:optimalg}
\end{figure}

    For MDP $\Mcal$, we choose reference state $0$ as $(0,1)$.
    A sufficient condition for $1$ and $2(b)$ to hold is the existence of a single stationary policy that induces an irreducible, ergodic Markov Chain, with the associated expected cost of first passage from any state $(x,y)$ to state $(0,1)$ being finite (\cite[Propositions 4 and 5]{sennott88}). Lemma~\ref{lemma:dtmcpolicy} verifies that this sufficient condition is met for our problem. 
    A sufficient condition for $2(a)$ is that $\Valp{s}$ is non-decreasing in $s$ \cite{sennott88}. 
    Proposition~\ref{prop:mp-mono} demonstrates that this sufficient condition is also met.  

Now, condition~$3$ of \Thmref{opt-policy-exist} asserts that under any $a$, the quantity $\sum_{s^\prime} \mathbb{P}_{s,s^\prime}(a)M(s^\prime)$ should be finite. For MDP $\Mcal$, from \eqnref{transproba0}, when $a=0$, we have for any state $s=(x,y)$,
\begin{equation}
    \sum_{s^\prime} \mathbb{P}_{s,s^\prime}(0)M(s^\prime) = p M(0,y+1) + \pbar M(x+1, y+1).
    \eqnlabel{Pssa0}
\end{equation}
From \eqnref{transproba1}, when $a=1$, we similarly have for any state $s=(x,y)$,
\begin{equation}
    \sum_{s^\prime} \mathbb{P}_{s,s^\prime}(1)M(s^\prime) = p M(0,x+1) + \pbar M(x+1, x+1).\eqnlabel{Pssa1}
\end{equation}
It follows from \eqnref{Pssa0}, \eqnref{Pssa1} and Lemma~\ref{lemma:dtmcpolicy} that condition $3$ holds for MDP $\Mcal$.
Therefore, there exists a constant $\gs = \lim_{\al \to 1}(1-\al)\Valp{x,y}$ for every state $(x,y)$ that is an optimal average cost and a relative cost to go function $f(x,y)$ with $0\leq f(x,y) \leq M(x,y)$. 


\section{Conclusion}
This paper focused on a class of systems where source updates are disseminated using shared memory. The Writer process records these source updates in the memory, and a Reader fulfills clients’ requests for these measurements by reading from the memory. 
We addressed optimizing memory access to minimize average cost, establishing the existence of an optimal stationary deterministic policy for our Markov Decision Process (MDP) and showing its threshold structure.

We found that in the optimal policy, if the Reader decides to stay idle when the memory undergoes an update, then the optimal action is to not sample in subsequent slots until the memory is updated with a fresh source update. 
Finally, we observe that the age achieved by this optimal policy serves as a lower bound on the achievable average age when the Reader is unaware of the memory state. 
Investigating the optimal policy structure for this scenario, where the Reader does not know when the memory is updated, is left for future research. 

\bibliographystyle{IEEEtran}
\bibliography{refs,AOI-2020-03}


\clearpage
\newpage

\appendix

\section*{Proof of Lemma~\ref{lemma:dtmcpolicy}}
\noindent
\textbf{Lemma~\ref{lemma:dtmcpolicy}.} Under the deterministic stationary policy of reading in every slot, the system exhibits an irreducible, ergodic Markov Chain, with expected cost $M(x,y)$ of first passage from state $s=(x,y)$ to $(0,1)$ satisfying
\begin{equation}
M(x,y)\le \frac{1+p}{p^2}(c+y) + \frac{3}{2p^3}.
\end{equation}

\begin{proof}
Note that \begin{equation}
\eqnlabel{Mxybound}
M(x,y)\le\E{\Chat(x,y)},
\end{equation} 
where $\Chat(x,y)$ is the first passage cost  under the policy in which the Reader samples in every slot.
Starting from state $(x,y)$ under the ``always sample'' policy, there is a geometric~$(p)$ number $N$ of slots in which the system passes from states $(x,y)$ up through $(x+N-1,y+N-1)$  until a memory update takes the system to state $(0,x+N)$. In the next slot, a cost $c+x+N$ is incurred and the system goes to either state $(0,1)$ with probability $p$ or, with probability $1-p$,  to $(1,1)$. In the latter case, the additional cost $\Chat(1,1)$ is incurred to reach $(0,1)$. We define the Bernoulli~$(1-p)$
 random variable $\overline{Z}$ such that $\overline{Z}=1$ if a memory update does {\em not} occur in state $(0,x+N)$.  The cost expended to go from $(x,y)$ to $(0,1)$ is then
\begin{IEEEeqnarray}{rCl}
    \Chat(x,y) &=& \sum_{j=y}^{\mathclap{y+N-1}}(c+j) + (c+x+N)
    +\overline{Z} \Chat(1,1)\eqnlabel{miieqn-1}\nn
&=& 
\scalebox{0.99}
{$
N(c+y) +(c+x)+ \frac{N(N+1)}{2}+ \overline{Z}\Chat(1,1). $}
\IEEEeqnarraynumspace
\end{IEEEeqnarray}
Taking expectation,
\begin{IEEEeqnarray}{rCl}
\E{\Chat(x,y)}
    &=& \frac{c +y}{p}
+ (c+x)+ \frac{3-p}{2p^2}+ \pbar \E{\Chat(1,1)}.\eqnlabel{miieqn-2}
\IEEEeqnarraynumspace 
\end{IEEEeqnarray}
Evaluating \eqnref{miieqn-2} at $(x,y)=(1,1)$ yields
\begin{equation}
    \E{\Chat(1,1)} = \frac{1}{p}\left[\left(\frac{1}{p} + 1\right)(c+1) + \frac{3-p}{2p^2}\right]. 
    \eqnlabel{m11eqn}
\end{equation}
Combining \eqnref{miieqn-2}
and \eqnref{m11eqn} yields 
\begin{IEEEeqnarray}{rCl}
    \E{\Chat(x,y)}
    &=& \frac{c+y}{p}
+ (c+x)+ \frac{1-p^2}{p^2}(c+1)+\frac{3-p}{2p^3}.
\IEEEeqnarraynumspace
\end{IEEEeqnarray}
Since $x\le y$ and $1\le y$ for any feasible state $(x,y)$, we obtain
\begin{IEEEeqnarray}{rCl}
    \E{\Chat(x,y)}
    &\le& \paren{\frac{1}{p}+1}(c+y)
+ \frac{1-p^2}{p^2}(c+y) +\frac{3-p}{2p^3}\nn
&\le& \frac{1+p}{p^2}(c+y) + \frac{3}{2p^3}.
\end{IEEEeqnarray}
The claim then follows from \eqnref{Mxybound}.
\end{proof}

\section*{Proof of 
Proposition~\ref{prop:mp-thresh}}
\label{proof:mp-thresh}
\noindent
\textbf{Proposition~\ref{prop:mp-thresh}.}
If the $\al$-optimal action is to sample in $(x,y)$, then the $\al$-optimal action is to sample in every $(x,y')$ with $y' \geq y$.
\begin{proof}
For brevity, we'll use the following shorthand notation in the proof.
For $w \leq v$, we define
\begin{equation}\eqnlabel{Jtildefn}
    \Jvt{u,v,w} = \Valp{u,v} - \Valp{u,w}.
\end{equation}
The monotonicity of the value function (Proposition~\ref{prop:mp-mono}) implies 
\begin{align}\eqnlabel{jtilvalmono}
    \Jvt{u,v_1,w} \leq \Jvt{u,v_2,w}\quad\text{for all $u$, $w$, and $v_1\le v_2$.}
\end{align}
For the rest of our discussion, we use the following form of discounted-cost Bellman's optimality equation with $c_\al = c/\alpha$:
\begin{align}
    \Valp{x,y} &= y+ \al\min\{p\Valp{0,y+1} + \pbar\Valp{x+1, y+1}, \nn
    & c_\al + p\Valp{0,x+1}+\pbar\Valp{x+1,x+1}\},
    \eqnlabel{mpvalfunc}
\end{align}
Let $\xhat=x+1$ and $\yhat=y+1$.
According to \eqnref{mpvalfunc}, the condition for the Reader to sample in $(x,y)$ is
\begin{equation}
p\Valp{0,\yhat} + \pbar\Valp{\xhat,\yhat} 
\geq c_\al + p\Valp{0,\xhat} + \pbar\Valp{\xhat,\xhat}. \eqnlabel{thresh-cond-1}
\end{equation}
Using the shorthand $\Jvt{u,v,w}$, 
the inequality \eqnref{thresh-cond-1} becomes
\begin{equation}
    p\Jvt{0,\yhat,\xhat} + \pbar \Jvt{\xhat,\yhat,\xhat} \geq c_\al. \eqnlabel{thresh-cond-2}
\end{equation}
Given that condition~\eqnref{thresh-cond-2} holds, we examine the state $(x,\yhat)$. The value function for this state is
\begin{align}
    \Valp{x,\yhat} = \yhat+ \al\min\{&p\Valp{0,\yhat+1} + \pbar\Valp{\xhat, \yhat+1}, \nn
&\phantom{=} c_\al + p\Valp{0,\xhat}+\pbar\Valp{\xhat,\xhat}\}. \eqnlabel{valfuncxyp1}
\end{align}
The condition for the Reader to sample in $(x,\yhat)$ is
\begin{equation}
    p\Valp{0,\yhat+1} + \pbar\Valp{\xhat, \yhat+1} \geq c_\al + p\Valp{0,\xhat}+\pbar\Valp{\xhat,\xhat},
\end{equation}
or equivalently,
\begin{equation}
    p \Jvt{0,\yhat+1,\xhat} + \pbar \Jvt{\xhat,\yhat+1,\xhat} \geq c_\al. 
    \eqnlabel{xyp1-cond}
\end{equation}
Now we observe from the monotonicity property \eqnref{jtilvalmono} and \eqnref{thresh-cond-2} that
\begin{align}
    p \Jvt{0,\yhat+1,\xhat} + \pbar \Jvt{\xhat,\yhat+1,\xhat} &\geq p\Jvt{0,\yhat,\xhat} + \pbar \Jvt{\xhat,\yhat,\xhat} \nn
    &\geq c_\al.
    \eqnlabel{verify-thresh-cond}
\end{align}
Thus \eqnref{xyp1-cond} holds, confirming  that the Reader samples in state $(x,\yhat)$. 
%
\end{proof}


\section*{Proof of Proposition~\ref{prop:concavity}}
\label{proof:concavity}
\noindent
\textbf{Proposition~\ref{prop:concavity}.}
(Concavity):
$\Valp{x,y}$ is concave in $y$. Specifically, for a fixed $x$, $\Valp{x,y+1} - \Valp{x,y}$ is non-increasing in $y$.
\begin{proof}
We want to show that for a fixed $x$,
$\Valpn{n}{x,i+1} - \Valpn{n}{x,i} \geq \Valpn{n}{x,i+2} - \Valpn{n}{x,i+1}$, for every $i \in \mathbb{N}$.
To achieve this, we focus on demonstrating the inequality:
\begin{equation}
    \Valpn{n}{x,i+2} + \Valpn{n}{x,i} \leq 2 \Valpn{n}{x,i+1} \qquad \forall~n,i. \eqnlabel{concave-cond}
\end{equation}
The base case for $n=1$ is trivially satisfied, as $\Valpn{1}{x,y} = y$. 
Now suppose \eqnref{concave-cond} holds for $n=1, 2 \ldots k$ for every $i$. We will establish the validity of \eqnref{concave-cond} under two scenarios, corresponding to the $\al$-optimal action at stage $k+1$ being either to sample or idle in state $(x,i+1)$. First, let's consider the case where it is optimal to sample in $(x,i+1)$ at stage $k+1$. This implies that the value iteration function in this state satisfies:
\begin{align}
    \Valpn{k+1}{x,i+1} &= i+1+c+\al p \Valpn{k}{0,x+1}\nn
    &\qquad + \al\pbar \Valpn{k}{x+1,x+1}. \eqnlabel{iplus1-sample}
\end{align}
Furthermore, leveraging Proposition \ref{prop:mp-thresh}, we deduce that sampling in $(x,i+1)$ is also the optimal action for state $(x,i+2)$ at stage $k+1$, resulting in:
\begin{align}
  \Valpn{k+1}{x,i+2} &= i+2+c+\al p \Valpn{k}{0,x+1}\nn
  &\qquad + \al\pbar \Valpn{k}{x+1,x+1}.\eqnlabel{iplus2-sample}  
\end{align}
    
Notice that the value iteration function for $(x,i)$ satisifies
\begin{equation}
    \scalebox{0.9}{
    $
    \Valpn{k+1}{x,i} \leq i+c+\al\left(p \Valpn{k}{0,x+1} + \pbar \Valpn{k}{x+1,x+1} \right). 
    $}
    \eqnlabel{i-sample}
\end{equation}
Combining \eqnref{iplus1-sample}, \eqnref{iplus2-sample}, and \eqnref{i-sample}, we establish:
\begin{equation}
    \Valpn{k+1}{x,i+2} + \Valpn{k+1}{x,i} \leq 2\Valpn{k+1}{x,i+1}.
\end{equation}
Let us now consider the situation where the $\al$-optimal action is to stay idle in state $(x,i+1)$ at stage $k+1$. This implies that
\begin{equation}
    \Valpn{k+1}{x,i+1} = i+1+\al\left(p \Valpn{k}{0,i+2} + \pbar \Valpn{k}{x+1, i+2}\right).\eqnlabel{iplus1-idle}
\end{equation}
Leveraging Proposition~\ref{prop:mp-thresh}, we conclude that staying idle in $(x,i+1)$ is also the optimal action for state $(x,i)$ at stage $k+1$, leading to:
\begin{equation}
    \Valpn{k+1}{x,i} = i+\al\left(p \Valpn{k}{0,i+1} + \pbar \Valpn{k}{x+1, i+1}\right).\eqnlabel{i-idle}
\end{equation}
The value iteration function for $(x,i+2)$ satisfies
\begin{equation}
    \Valpn{k+1}{x,i+2} \leq i+2+\al\left(p \Valpn{k}{0,i+3} + \pbar \Valpn{k}{x+1, i+3}\right).\eqnlabel{iplus2-idle}
\end{equation}
Combining \eqnref{i-idle} and \eqnref{iplus2-idle}, we can demonstrate:
\begin{align}
    &\Valpn{k+1}{x,i+2} + \Valpn{k+1}{x,i} \nn
    &\leq 2(i+1) + \al[p \left(\Valpn{k}{0,i+3} + \Valpn{k}{0,i+1}\right) 
    \nn
    &\qquad + \pbar\left(\Valpn{k}{x+1, i+3}+\Valpn{k}{x+1, i+1}\right)], 
    \nn
    &\overset{(a)}\leq 2(i+1) + \al\left[2p\Valpn{k}{0,i+2} + 2\pbar\Valpn{k}{x+1,i+2}\right], \nn
    &= 2\left(i+1 + \al\left[p\Valpn{k}{0,i+2} + \pbar\Valpn{k}{x+1,i+2}\right)\right], \nn
    &\overset{(b)}=2\Valpn{k+1}{x,i+1},
\end{align}
where $(a)$ follows from induction hypothesis that $\Valpn{k}{x,i+3} + \Valpn{k}{x,i+1} \leq 2 \Valpn{k}{x,i+2}$, and $(b)$ follows from \eqnref{iplus1-idle}. 
It follows from principle of mathematical induction that \eqnref{concave-cond} holds for every $n$, and hence $\Valpn{n}{x,y}$ is concave in $y$.
As $\lim_{n\to\infty}\Valpn{n}{x,y} = \Valp{x,y}$, this implies that $\Valp{x,y}$ is concave in y.
\end{proof}

\section*{Proof of Proposition~\ref{prop:policy-hypo}}
\label{proof:policy-hypo}
\noindent
\textbf{Proposition~\ref{prop:policy-hypo}.}
If the $\al$-optimal action in state $(x,y)$ is to idle, then the $\al$-optimal action in states $(x+i,y+i), \forall i\geq 1$ is to stay idle.
\begin{proof}
For brevity, we'll use the following shorthand notation in the proof.
For $w \leq v$, let 
\begin{equation}
    \Jt{n}{u,v,w} = \Valpn{n}{u,v} - \Valpn{n}{u,w}.
\end{equation}
We establish key properties of $\Jt{n}{u,v,w}$ to be utilized later in the proof.
\begin{enumerate}
    \item 
        Given $w\leq v$, Proposition~\ref{prop:mp-mono} implies $\Valpn{n}{u,v} \geq \Valpn{n}{u,w} \geq 0$ and hence $\Jt{n}{u,v,w} \geq 0$.
        
    \item
        If $v_2 \geq v_1$, and $w_2 \geq w_1$, it follows from concavity property (Proposition~\ref{prop:concavity}) that
        \begin{equation}
            \Valpn{n}{u,v_2} - \Valpn{n}{u,w_2} \leq \Valpn{n}{u,v_1} - \Valpn{n}{u,w_1}
        \end{equation}
        and as a consequence,
        \begin{equation}
        \scalebox{0.9}{
        $
        \Jt{n}{u,v_2,w_2} \leq \Jt{n}{u,v_1,w_1}, \forall\text{$u$, $w_1 \le w_2$, and $v_1\le v_2$}.
        $}  
        \eqnlabel{jtilmono}
        \end{equation}

    \item 
        Let $\uhat = u+1$, $\vhat = v+1$ and $\what=w+1$.
        Under the condition of not sampling in $(u,v)$, it can be shown that
        \begin{equation}
        \scalebox{0.9}{
        $
            \Jt{n}{u,v,w} = v-w + \al(p\Jt{n-1}{0,\vhat,\what} + \pbar\Jt{n-1}{\uhat,\vhat,\what}). 
        $}    
        \eqnlabel{jtiliter}
        \end{equation} 
\end{enumerate}

We now resume the proof of proposition. Letting $\xhat = x+1$ and $\yhat = y+1$ and $c_\al = c/\al$, we re-write the value iteration in state $(x,y)$ given by \eqnref{valiteropteqn} as:
\begin{align}
    \Valpn{n+1}{x,y} &= y + \al\min\{ p\Valpn{n}{0,\yhat} + \pbar\Valpn{n}{\xhat, \yhat}, \nn
                     &\phantom{=}\qquad\qquad c_\al + p\Valpn{n}{0,\xhat}+\pbar\Valpn{n}{\xhat,\xhat}\},
    \eqnlabel{valiterxy}
\end{align}
Given that Reader doesn't sample in $(x,y)$ implies that for all $n$, the terms inside the $\min$ function in \eqnref{valiterxy} satisfy:
\begin{equation}
    p\Valpn{n}{0,\yhat} + \pbar\Valpn{n}{\xhat,\yhat} 
    \leq c_\al + p\Valpn{n}{0,\xhat} + \pbar\Valpn{n}{\xhat,\xhat}.
    \eqnlabel{given-cond-1}
\end{equation}
Expressing inequality \eqnref{given-cond-1} in terms of $\Jt{n}{u,v,w}$, we get:
\begin{equation}
    p\Jt{n}{0,\yhat,\xhat} + \pbar \Jt{n}{\xhat,\yhat,\xhat} \leq c_\al.
    \eqnlabel{given-cond-2}
\end{equation}
Given that \eqnref{given-cond-2} holds for every $n$,
we examine
state $(x+i,y+i)$. The value iteration expression at stage $n+1$ is given by:
\begin{align}
    &\hspace{-1.2em}\Valpn{n+1}{x+i,y+i} \nn
    &= y+i + \al\min\{ p\Valpn{n}{0,\yhat+i} + \pbar\Valpn{n}{\xhat+i, \yhat+i}, \nn
    &\phantom{=}\qquad\qquad c_\al + p\Valpn{n}{0,\xhat+i}+\pbar\Valpn{n}{\xhat+i,\xhat+i}\}.
    \eqnlabel{iyivaliter}
\end{align}
To establish that the optimal action in state $(x,y)$ being to stay idle implies the same for states $(x+i,y+i)$, we aim to show that the terms inside the $\min$ function in \eqnref{iyivaliter} satisfy:
\begin{align}
    p\Valpn{n}{0,\yhat+i} + \pbar\Valpn{n}{\xhat+i, \yhat+i} 
    &\leq c_\al + p\Valpn{n}{0,\xhat+i}
    \nn
    &\qquad +\pbar\Valpn{n}{\xhat+i,\xhat+i}.
    \eqnlabel{prov-cond-1}
\end{align}
or equivalently,
\begin{equation}
    p\Jt{n}{0,\yhat+i,\xhat+i} + \pbar \Jt{n}{\xhat+i, \yhat+i, \xhat+i} \leq c_\al.
\eqnlabel{prov-cond-2}
\end{equation}
Given that \eqnref{given-cond-2} holds for all $n$, proving that \eqnref{prov-cond-2} holds for all $n$ is equivalent to showing that the LHS of \eqnref{prov-cond-2} is less than LHS of \eqnref{given-cond-2}. 
For that it is sufficient to show for all $n\geq1$
\begin{equation}
    I_1(n+1) = \Jt{n}{0,\yhat+i,\xhat+i} - \Jt{n}{0,\yhat,\xhat} \leq 0, \eqnlabel{i1eqn}
\end{equation}
and 
\begin{equation}
    I_2(n+1) = \Jt{n}{\xhat+i, \yhat+i, \xhat+i} - \Jt{n}{\xhat,\yhat,\xhat} \leq 0. \eqnlabel{i2eqn}
\end{equation}
With $i \geq 1$, it is clear that $\yhat+i \geq \yhat$ and $\xhat+i \geq \xhat$. Leveraging 
\eqnref{jtilmono}
, we conclude that $\Jt{n}{0,\yhat+i,\xhat+i} \leq \Jt{n}{0,\yhat,\xhat}$, leading to $I_1 \leq 0$ for every $n$.
We use inductive arguments to show that $I_2(n+1) \leq 0$.
When $n=1$, we see that
\begin{equation}
    \Jt{1}{u,v,w} = \Valpn{1}{u,v} - \Valpn{1}{u,w} = v-w.
\end{equation}
This means that 
$$I_2(2) = \Jt{1}{\xhat+i, \yhat+i, \xhat+i} - \Jt{1}{\xhat,\yhat,\xhat} = 0, $$
and hence the base case holds.
Now assume that $I_2(n+1) \leq 0$ for $n=1,\ldots k-1$ for all $i \geq 0$. This implies:
\begin{equation}
    I_2(k) = \Jt{k-1}{\xhat+i, \yhat+i, \xhat+i} - \Jt{k-1}{\xhat,\yhat,\xhat} \leq 0, \quad \forall i \geq 0. \eqnlabel{i2eqnatk}
\end{equation}
We have established that $I_1(k) \leq 0$, implying that \eqnref{i1eqn} holds at $n=k-1$. Combining this with \eqnref{i2eqnatk}, we conclude that both $I_1$ and $I_2$ hold at $n=k-1$. This, in turn, implies that \eqnref{prov-cond-2} holds at $n=k-1$. Consequently, \eqnref{prov-cond-1} holds at $n=k-1$. 
Therefore,
the assumption $I_2(k) \leq 0$ for all $i \geq 0$ implies that the action that minimizes \eqnref{iyivaliter} at stage $k$ is to stay idle in state $(x+i,y+i)$ for all $i \geq 1$.

Now, we need to demonstrate that: 
\begin{equation}
    I_2(k+1) = \Jt{k}{\xhat+i, \yhat+i, \xhat+i} - \Jt{k}{\xhat,\yhat,\xhat} \leq 0, \quad \forall i \geq 0.
\end{equation}
Given that the assumption is to not sample in $(\xhat+i, \yhat+i)$ for $i \geq 0$ at stage $k$, this means that it is optimal to not sample in $(\xhat,\yhat)$ (Proposition~\ref{prop:mp-thresh}). Hence, employing Property $(3)$ of $\Jt{n}{u,v,w}$, from 
\eqnref{jtiliter}, we have with $\ihat=i+1$,
\begin{align}
    &\hspace{-5em}\Jt{k}{\xhat+i, \yhat+i, \xhat+i} \nn
    &= \yhat-\xhat + \al(p\Jt{k-1}{0, \yhat+\ihat, \xhat+\ihat} 
    \nn 
    &\qquad + \pbar\Jt{k-1}{\xhat+\ihat, \yhat+\ihat, \xhat+\ihat}).
\end{align}
Similarly, we have
\begin{align}
    \Jt{k}{\xhat, \yhat, \xhat}
    &= \yhat-\xhat  + \al(p\Jt{k-1}{0, \yhat+1, \xhat+1} 
    \nn
    &\qquad + \pbar\Jt{k-1}{\xhat+1, \yhat+1, \xhat+1}).
\end{align}
From \eqnref{jtilmono}, we can state that:
\begin{equation}
\Jt{k-1}{0, \yhat+\ihat, \xhat+\ihat} \leq \Jt{k-1}{0, \yhat+1, \xhat+1}. \eqnlabel{hypo-tmp1}
\end{equation}
Additionally, it follows from \eqnref{i2eqnatk},
\begin{equation}
\Jt{k-1}{\xhat+\ihat, \yhat+\ihat, \xhat+\ihat} \leq \Jt{k-1}{\xhat+1, \yhat+1, \xhat+1}. \eqnlabel{hypo-tmp2}
\end{equation}
Based on \eqnref{hypo-tmp1} and \eqnref{hypo-tmp2}, we observe that
\begin{equation}
    \Jt{k}{\xhat+i, \yhat+i, \xhat+i} - \Jt{k}{\xhat, \yhat, \xhat} = I_2(k+1) \leq 0.
\end{equation}
Thus, by induction, we establish that $I_2(n+1) \leq 0$ holds for all $n \geq 1$. 
\end{proof}




\section*{Proof of Proposition~\ref{prop:rel-cost-func}}
\label{proof:rel-cost-func}
\noindent
\textbf{Proposition~\ref{prop:rel-cost-func}.}
Defining $(0,1)$ as the reference state with $f(0,1) = 0$, the relative cost functions satisify:
\begin{enumerate}[label={(\roman*)},itemindent=1em]

    
    \item 
        \begin{equation}
            f(0,\yzero+1) - f(0,\yzero) = 1.
        \end{equation}
        \begin{proof} 
        Given that it is optimal to sample in $(0,\yzero)$, the relative-cost Bellman's equation in state $(0,\yzero)$ is given as
        \begin{equation}
            \gs + f(0,\yzero) = \yzero + c + pf(0,1) + \pbar f(1,1). \eqnlabel{be-0y0}
        \end{equation}
        The optimal action in $(0,\yzero+1)$ is also to sample (Proposition~\ref{prop:mp-thresh}), and therefore, the relative-cost Bellman's equation in state $(0,\yzero+1)$ becomes 
        \begin{equation}
            \gs + f(0,\yzero+1) = \yzero + 1 + c + pf(0,1) + \pbar f(1,1). \eqnlabel{be-0y01}
        \end{equation}
        It follows from \eqnref{be-0y0} and \eqnref{be-0y01} that
        $f(0,\yzero+1) - f(0,\yzero) = 1$.
        \end{proof}
    
    \item 
        For any $x \geq 0$,
        \begin{equation}
            f(x, \yzero-1) = \frac{1}{p}(\jzero + \frac{\pbar}{p}) - 1,
        \end{equation}
        where $\jzero = \yzero - g + pf(0,\yzero)$. 
        \begin{proof}
            For any $x\geq 0$, the optimal action in $(x,\yzero-1)$ is to idle, and the Bellman's equation \eqnref{rc-bellmans} becomes
            \begin{equation}
                f(x,\yzero-1) = -\gs + \yzero-1 + p f(0,\yzero) + \pbar f(x+1,\yzero).
            \end{equation}
            Let $\jzero = -\gs + \yzero + p f(0,\yzero)$, we obtain
            \begin{equation}
                f(x,\yzero-1) = \jzero  - 1 + \pbar f(x+1,\yzero).
                \eqnlabel{xy0min1-1}
            \end{equation}
            Since the optimal action in $(x,\yzero-1)$ is to idle, then from Proposition~\ref{prop:policy-hypo}, the optimal action in $x \geq 0$, $(x+1,\yzero), $ is to idle as well. The Bellman's equation \eqnref{rc-bellmans} in $(x+1,\yzero)$ becomes
            \begin{align}
                &f(x+1,\yzero) \nn
                &= -\gs + \yzero + p f(0,\yzero+1) + \pbar f(x+2,\yzero+1), \nn
                &\overset{(a)}=  -\gs + \yzero + p(1+f(0,\yzero)) + \pbar f(x+2,\yzero+1), \nn
                &= \jzero + p + \pbar f(x+2,\yzero+1),
                \eqnlabel{xp1y0}
            \end{align}
            where $(a)$ follows from Proposition~\ref{prop:rel-cost-func}(i). Substituting \eqnref{xp1y0} into \eqnref{xy0min1-1}, we obtain
            \begin{equation}
                f(x,\yzero-1) = \jzero(1+\pbar) - 1 + p \pbar + \pbar^2 f(x+2,\yzero+1).
                \eqnlabel{y0min1-2}
            \end{equation}
            Repeating this procedure $n$ times yields
            \begin{align}
                f(x,\yzero-1) &= \jzero\sum_{i=0}^{n}\pbar^i + p\pbar \sum_{i=1}^{n-1}(i+1)\pbar^i 
                \nn
                &+ \pbar^2 \sum_{i=0}^{n-2}(i+1)\pbar^i 
                \nn
                &+ \pbar^{n+1}f(x+n+1, \yzero+n) - 1,
            \end{align}
            and in the limit $n \to \infty$ we have
            \begin{align}
                f(0,\yzero-1) &= \frac{\jzero}{1-\pbar} + \frac{p \pbar}{(1-\pbar)^2} + \frac{\pbar^2}{(1-\pbar)^2} - 1, \nn
                &= 
                \frac{1}{p}(\jzero + \frac{\pbar}{p}) - 1.
            \end{align}
            Here, $\pbar^{n+1}f(x+n+1, \yzero+n) \to 0$ when $n \to \infty$ as $f(x+n+1, \yzero+n)$ is bounded. This bounding property is derived from \Thmref{opt-policy-exist}, where it is established that $f(x+n+1, \yzero+n) \leq M(x+n+1, \yzero+n)$. 
            Then it follows from \eqnref{mxybound},
            \begin{equation}
                f(x+n+1,\yzero+n) \le\frac{1+p}{p^2}(c+\yzero+n) + \frac{3}{2p^3}.
            \end{equation}


        \end{proof}

    \item 
        For every $y < \yzero$, 
        \begin{equation}
            f(0,y) = f(1,y)= \cdots = f(y,y).
        \end{equation}
        \begin{proof}
            From the threshold structure of the optimal policy, the optimal action in $(x,\yzero-2)$ is to stay idle, and the relative-cost Bellman's equation \eqnref{rc-bellmans} becomes
            \begin{align}
                &f(0,\yzero-2) \nn
                &= -\gs + \yzero-2+p f(0,\yzero-1) + \pbar f(1,\yzero-1), \nn
                &\overset{(a)}= -\gs + \yzero-2+p f(0,\yzero-1) + \pbar f(0,\yzero-1), \nn
                &= -\gs + \yzero-2+ f(0,\yzero-1),
                \eqnlabel{xy0min2-1}
            \end{align}
            where $(a)$ follows from Proposition~\ref{prop:rel-cost-func}(ii) as $f(x,\yzero-1)$ is independent of $x$.
            This fact along with \eqnref{xy0min2-1} implies that $f(x,\yzero-2)$ is also independent of $x$, and so $f(0,\yzero-2) = f(1,\yzero-2) \ldots f(\yzero-2,\yzero-2)$.
            In fact this can be generalized such that $(x,\yzero-k)$ with $x \geq 0$ and $k \in \{1, 2, \ldots, \yzero-1\}$ is independent of $x$.
        \end{proof}
    
        
    \item  
        When $\yzero > 1$, 
        \begin{equation}
        f(0,\yzero) = \yzero-\gs + c.
        \end{equation}
        \begin{proof}
            At $(0,\yzero)$ the Reader samples and the Bellman's equation \eqnref{rc-bellmans} becomes
            \begin{equation}
                f(0,\yzero) = \yzero - \gs + c + p f(0,1) + \pbar f(1,1).
            \end{equation}
            When $\yzero > 1$, we have from Proposition~\ref{prop:rel-cost-func}\ref{prop:rc-leq-y0}, $f(0,1) = f(1,1)$, and since $f(0,1) = 0$, it follows that
            \begin{equation}
                f(0,\yzero) = \yzero - \gs + c.
            \end{equation}
        \end{proof}

        
\end{enumerate}

\section*{Proof of Lemma~\ref{lemma:gy0exp}}
\label{proof:gy0exp}
\noindent
\textbf{Lemma~\ref{lemma:gy0exp}.}
The average cost as a function of the threshold $\yzero$ is given by:
\begin{equation}
    \gzero{\yzero} = \frac{1}{2}\left(\frac{1}{p} + \yzero + \frac{2cp + \pbar/p}{p\yzero+\pbar}\right).
\end{equation}
\begin{proof}
We break down the proof into three parts. In the first and second parts, we derive analytical expressions for the optimal average cost when $\yzero=1$ and $\yzero=2$, respectively. In the third part, we focus on obtaining a general expression for the optimal average cost when $\yzero>2$. Surprisingly, we discover that the average cost equation as a function of $\yzero$ obtained in the third part is a general equation for any $\yzero \geq 1$.

\medskip
\noindent
(\textit{Part $1$}):    If $\yzero=1$, it is optimal to sample in $(0,1)$. Thus  Bellman's equation \eqnref{rc-bellmans} yields
\begin{align}
 f(0,1) &= -\gs + 1 + c + pf(0,1) + \pbar f(1,1).
 \end{align}
 Defining $(0,1)$ as the reference state with $f(0,1)=0$ yields
   $0 = -\gs +1 +c+\pbar f(1,1)$, or equivalently,
\begin{align}
        f(1,1)&= \frac{\gs - 1 - c}{\pbar}.
        \eqnlabel{eqn11}
    \end{align}
    Now it is optimal to never sample in $f(1,1)$. Then
    \begin{align}
        f(1,1) &= -\gs + 1 + p f(0,2) + \pbar f(2,2), \nn
        &\overset{(a)}= \gs + 1 + p (1+f(0,1)) + \pbar f(2,2),\nn
        &= -\gs + 1 + p + \pbar f(2,2),
        \eqnlabel{eqn11-1}
    \end{align}
    where $(a)$ follows from Proposition \ref{prop:rel-cost-func}\ref{prop:rc-diff-y0y0pl1}.
    Since staying idle is the optimal action in $(1,1)$, then Proposition~\ref{prop:policy-hypo} implies that staying idle is also the optimal action in state $(2,2)$, and hence
    \begin{align}
        f(2,2) &= -\gs + 2 + p f(0,3) + \pbar f(3,3), \nn
        &= -\gs + 2 + p (2 + f(0,1)) + \pbar f(3,3), \nn
        &= -\gs + 2 + 2p + \pbar f (3,3).
        \eqnlabel{eqn22}
    \end{align}
    Substituting $f(2,2)$ obtained in \eqnref{eqn22} in \eqnref{eqn11-1}, we obtain
    \begin{equation}
        f(1,1) = -\gs(1+\pbar) + 1 + p + 2\pbar+2p\pbar+\pbar^2f(3,3). 
    \eqnlabel{eqn11-2}
    \end{equation}
    Similar to above, we can obtain $f(3,3)$ as
    \begin{equation}
        f(3,3) = -\gs + 3 + 3p + \pbar f(4,4).
        \eqnlabel{eqn33}
    \end{equation}
    Again substituting $f(3,3)$ obtained in \eqnref{eqn33} in \eqnref{eqn11-2}, we obtain
    \begin{align}
        f(1,1) &= -\gs(1+\pbar+\pbar^2) + 1+p + 2\pbar + 2p\pbar + \nn
        &\qquad\qquad 3\pbar^2 + 3p\pbar^2 +\pbar^3f(4,4). 
    \end{align}
    Repeating this $n$ times we obtain
    \begin{align}
        f(1,1) &= -\gs(1+\pbar+\pbar^2 \ldots +\pbar^n)  \nn
        &+(1+2\pbar+3\pbar^2+\ldots+(n+1)\pbar^n) \nn
        &+p(1+2\pbar+3\pbar^2+\ldots+(n+1)\pbar^n) \nn
        &+ \pbar^{n+1} f(n+2, n+2),
    \end{align}
    and letting $n\to\infty$, we obtain
    \begin{align}
        f(1,1) &= \frac{-\gs}{1-\pbar} + \frac{1}{(1-\pbar)^2} + \frac{p}{(1-p)^2}, \nn
        &= \frac{-\gs}{p} + \frac{1}{p^2} + \frac{1}{p}, \nn
        &= \frac{-\gs+1}{p} + \frac{1}{p^2}. \eqnlabel{eqn11-n}
    \end{align}
    Here, $\pbar^{n+1} f(n+2, n+2) \to 0$ when $n \to \infty$ as $f(n+2, n+2)$ is bounded. This bounding property is derived from \Thmref{opt-policy-exist}, where it is established that $f(n+2, n+2) \leq M(n+2, n+2)$. Subsequently, \eqnref{mxybound} implies that
    \begin{equation}
        f(n+2, n+2) \leq \frac{1+p}{p^2}(c+n+2) + \frac{3}{2p^3}.
    \end{equation}
    Now equating $f(1,1)$ in  \eqnref{eqn11} and \eqnref{eqn11-n}, we obtain
    \begin{align}
    g &= \frac{1}{p} + cp. \eqnlabel{gaty0eq1}
    \end{align}

\medskip
\noindent
(\textit{Part $2$}): If $\yzero=2$, then it is optimal to sample in $(0,2)$. The realtive-cost Bellman's equation \eqnref{rc-bellmans} is
    \begin{align}
        f(0,2) &= -\gs + 2 + c + p f(0,1) + \pbar f(1,1), \nn
        &\overset{(a)}= -\gs + 2 + c, \eqnlabel{eqn02}
    \end{align}
    where $(a)$ follows from Proposition~\ref{prop:rel-cost-func}\ref{prop:rc-leq-y0} which for $\yzero=2$ implies that $f(1,1) = f(0,1) = 0$.
    Now from \eqnref{rc-y0min1}, we have for $x = 1$, and $\yzero=2$,
    \begin{align}
        f(0,1) &= \frac{1}{p}\left(\jzero+\frac{\pbar}{p}\right) - 1, \nn
        &= \frac{1}{p}\left(-\gs + 2 + p f(0,2) + \frac{\pbar}{p}\right) - 1.
    \end{align}
    With $f(0,2)$ given by \eqnref{eqn02}, we have
    \begin{equation}
        f(0,1)=\frac{1}{p}\left(-\gs + 2 + p(-\gs + 2 + c) + \frac{\pbar}{p}\right) - 1.
    \end{equation} 
    Equating $f(0,1)=0$ gives
    \begin{align}
        \gs &= \frac{1}{1+p}\left(2 + 2p + cp + \frac{\pbar}{p} - p\right),\nn
        &=\frac{1}{1+p}\left( 1+ 1 + p + cp + \frac{\pbar}{p}\right),\nn
        &= \frac{1}{1+p}\left(1 + \frac{1}{p} + p (c+1)\right), \nn
        &= \frac{1}{p} + \frac{(c+1)p}{1+p}. \eqnlabel{gaty0eq2}
    \end{align}

\medskip
\noindent
(\textit{Part $3$}): Now consider the case when $\yzero > 2$. Since it is optimal to not sample in state $(0,\yzero-2)$, the Bellman's equation for state $(0,\yzero-2)$ becomes
    \begin{align}
        f(0,\yzero-2) &= \yzero-2 -\gs  +p f(0,\yzero-1) + \pbar f(1,\yzero-1), \nn
        &\overset{(a)}= \yzero-2 -\gs+ f(0,\yzero-1),
    \end{align}
    where $(a)$ follows from  
    Proposition~\ref{prop:rel-cost-func}\ref{prop:rc-leq-y0}.
    Moreover,
    \begin{align}
        f(0,\yzero-3) &= \yzero-3 -\gs + p f(0,\yzero-2) + \pbar f(0,\yzero-2),
        \nn
        &= \yzero-3 -\gs + f(0,\yzero-2), \nn
        &= 2(\yzero-\gs)- (2+3) + f(0,\yzero-1).
    \end{align}
    Repeating this procedure $k$ times yields,
    \begin{align}
        &f(0,\yzero-k) \nn
        &= (k-1)(\yzero-\gs) - (2+3+4+\dots+k) + f(0,\yzero-1), \nn
        &= (k-1)(\yzero-\gs) - \frac{k(k+1)}{2} +1 + f(0,\yzero-1).
        \eqnlabel{yzero-k}
    \end{align}
    Recalling $(0,1)$ as the reference state with $f(0,1)=0$, evaluating \eqnref{yzero-k} at $k=\yzero-1$ yields
    \begin{equation}\eqnlabel{yzero-ksub}
        (\yzero-2)(\yzero-\gs) = \frac{(\yzero-1)\yzero}{2} +1+f(0,\yzero-1),
    \end{equation}
    From Proposition~\ref{prop:rel-cost-func},
    \begin{IEEEeqnarray}{rCl}
        f(0,\yzero-1)&=&(1+1/p)(\yzero-g)+c+\pbar/p^2-1.\IEEEeqnarraynumspace
    \end{IEEEeqnarray}
    Thus it follows from \eqnref{yzero-ksub} that
    \begin{equation}
        (\yzero-\gs)\left(\frac{1}{p} + \yzero - 1\right) = \frac{(\yzero-1)\yzero}{2} - \frac{\pbar}{p^2} - c.
    \end{equation}
    Rearranging to solve for $\gs$ yields
    \begin{align}
        \gs &= \yzero-\frac{1}{\yzero-1+1/p}\left[\frac{\yzero(\yzero-1)}{2} - c - \frac{\pbar}{p^2}\right]\nn
        &=\frac{\yzero}{2}+\frac{\frac{\yzero(\yzero-1+1/p)}{2}-\left[\frac{\yzero(\yzero-1)}{2} - c - \frac{\pbar}{p^2}\right]}{\yzero-1+1/p}\nn
        &=\frac{\yzero}{2}+\frac{\frac{\yzero}{2p}+ c +\frac{\pbar}{p^2}}{\yzero-1+1/p}.
    \end{align}
    Recalling $-1+1/p= \pbar/p$, we obtain
    \begin{align}
         \gs &= \frac{\yzero}{2} + \frac{1}{2p} + \frac{1}{2}\frac{2cp + \pbar/p}{p\yzero + \pbar}. \eqnlabel{gy0final-1}
    \end{align}
Since \eqnref{gy0final-1} depends upon $\yzero$, we express it as
\begin{equation}
    \gzero{\yzero} = \frac{1}{2}\left(\frac{1}{p} + \yzero + \frac{2cp + \pbar/p}{p\yzero+\pbar}\right). \eqnlabel{gy0final-2}
\end{equation}
Finally observe that even though \eqnref{gy0final-2} was derived for $\yzero > 2$, we see that 
$\gzero{1} = 1/p + cp$, where the RHS is same as RHS of \eqnref{gaty0eq1}, the average cost obtained separately at $\yzero=1$. Similarly, $\gzero{2} = 1/p + (c+1)p/(1+p)$, where the RHS is same as RHS of \eqnref{gaty0eq2}, the average cost obtained separately at $\yzero=2$.

\end{proof}

\end{document}